\documentclass[aip,jmp,reprint]{revtex4-2}

\usepackage{amsmath,amssymb,amsfonts,amsthm}
\usepackage{bm}
\usepackage{color}
\usepackage{epsfig}
\usepackage{graphicx}

\usepackage[utf8]{inputenc}
\usepackage[T1]{fontenc}
\usepackage{mathptmx}

\newtheorem{theorem}{Theorem}[section]

\newtheorem{lemma}[theorem]{Lemma}
\newtheorem{proposition}[theorem]{Proposition}

\numberwithin{equation}{section}

\newcommand{\CC}{{\mathbb C}}
\newcommand{\RR}{{\mathbb R}}

\newcommand{\NN}{{\mathbb N}}

\newcommand{\cF}{{\mathcal{F}}}
\newcommand{\cH}{{\mathcal{H}}}

\newcommand{\cR}{{\mathcal{R}}}
\newcommand{\cS}{{\mathcal{S}}}

\newcommand{\bp}{\bm{p}}

\newcommand{\bx}{\bm{x}}

\newcommand{\bO}{\bm{O}}

\newcommand{\bP}{\bm{P}}
\newcommand{\bQ}{\bm{Q}}

\newcommand{\fA}{{\mathfrak A}}

\newcommand{\fR}{{\mathfrak R}}

\newcommand{\bpartial}{\bm{\partial}}

\def\ie{{\it i.e.\ }}
\def\viz{{\it viz.\ }}
\def\etc{{\it etc}}

\newcommand{\be}{\begin{equation}}
\newcommand{\ee}{\end{equation}}

\begin{document}

\title{\Large Proper condensates} 



\author{Detlev Buchholz}
\email[]{detlev.buchholz@mathematik.uni-goettingen.de}
\affiliation{Mathematisches Institut, Universit\"at G\"ottingen,
  B\"urgerstraße 40, 37073 G\"ottingen - Germany}



\date{9 September 2021}

\begin{abstract} \noindent
  In this article a
  novel characterization of Bose-Einstein condensates is proposed.
  Instead of relying on occupation numbers of a few dominant modes,
  which become macroscopic in the limit of
  infinite particle numbers, 
  it focuses on the regular excitations whose numbers stay
  bounded in this limit. In this manner, subspaces of global, 
  respectively local regular wave functions are identified.
  Their orthogonal complements determine the wave functions
  of particles forming proper (infinite) condensates in the limit.
  In contrast to the concept
  of macroscopic occupation numbers, which does not sharply fix
  the wave functions of condensates in the limit states,
  the notion of proper condensates is unambiguously defined. 
  It is outlined, how this concept can be used in the
  analysis of condensates in models. The method is illustrated
  by the example of trapped non-interacting ground states and
  their multifarious thermodynamic limits, differing by
  the structure of condensates accompanying the Fock vacuum. 
  The concept of proper condensates is also compared
  with the Onsager-Penrose criterion, based on the analysis of 
  eigenvalues of one-particle density matrices. It is
  shown that the concept of regular wave functions is 
  useful there as well for the identification of 
  wave functions forming proper condensates. 
  \end{abstract}


\maketitle


%
\vspace*{-4mm}
\section{Introduction}
\setcounter{equation}{0}  

\noindent We propose in this article a novel characterization of
Bose-Einstein condensates. It differs from various previous
approaches, based on concepts such as the
spectrum of one-particle
density matrices in approximating states, the 
spontaneous breakdown of gauge symmetries in the
thermodynamic limit, or the appearance of long range order, 
cf.\ for example \cite{LiSeSoYn, Ve}.
The latter concepts focus on global properties of the
underlying systems. On the other hand, 
condensates are prepared in realistic experiments
in bounded regions with a limited number of particles.

\medskip 
The primary problem appearing in the formulation of
a corresponding \textit{local} criterion for condensation 
is due to the fact that the states of interest, 
having a finite particle number, can be  
described in the Fock representation (they are \textit{locally normal}).
So a sharp criterion that characterizes condensates
locally seems to be out of reach. In order to overcome this
difficulty, we consider the idealization of
an unlimited number of particles occupying 
regions, which are suitably adjusted to the
particle numbers. If the resulting 
states separate in the limit in an unambiguous manner into
an infinite number of particles, occupying a few 
states, and a regular component
consisting of finite numbers of particles occupying each of the remaining
states, we speak of \textit{proper condensation}. As we shall see,
the determination of the regular components is
crucial for the identification of the 
wave functions of the condensate. The 
idealizations underlying our approach are out of experimental reach, but
they allow it on the theoretical side to
determine in a clear-cut manner the wave functions
of particles forming condensates.
Once this has been accomplished, one can 
return to the states with a large, but finite particle number
and study the onset of proper condensation, depending on
data such as the temperature, the shape of trapping
potentials, \etc. 

\medskip
We are interested in states containing a finite number
of particles in $s$-dimensio\-nal space $\RR^s$, which are
confined by a trapping potential. Since we need to proceed 
to limits involving an infinite number of particles, including
the passage to appropriate thermodynamic limits, it is 
convenient to describe the states by positive linear
functionals on a specific algebra of bounded operators,
the resolvent algebra $\fR(\RR^s)$ 
introduced in \cite{BuGr}. It is superior to the Weyl algebra
since contributions due to
infinite accumulations of particles are effectively
suppressed, whereas on the Weyl algebra they lead to
singular states which often defy a meaningful physical interpretation.

\medskip
The resolvent algebra contains a 
subalgebra $\fA(\RR^s)$ of observables which 
do not change particle numbers \cite{Bu}.
Observables which are localized in open bounded or unbounded regions
$\bO \subset \RR^s$ are described by subalgebras
$\fA(\bO) \subset \fA(\RR^s)$.
What matters here is the fact that these algebras
contain the
resolvents of corresponding particle number operators,
\mbox{$\mu \mapsto (\mu 1 + a^*(f) a(f))^{-1}$}, for
$f \in L^2(\bO)$ and $\mu > 0$, where
$a^*(f), a(f)$ are creation and annihilation
operators. Our criterion, characterizing
states containing proper condensates, is based on these
operators. 

\medskip
In the subsequent section we define the notion of proper
condensates, exhibit some of its features, and indicate
how it can be used in applications. In Sect.\ 3 we 
illustrate the method by an analysis of the thermodynamic
limits of non-interacting bosons that occupy
the ground states in regular 
trapping potentials. Depending on how this limit
is reached, the resulting Fock vacuum is accompanied by 
different arrangements of condensates. Sect.\ 4 contains a 
study of the relation between the notion of
proper condensates and the concept of
macroscopic occupation numbers, invented by
Onsager and Penrose \cite{OnPe}, which is based on the
analysis of one-particle density matrices. Our article concludes with
an outlook on further applications of our framework. 
In an appendix some general properties of the 
occupation numbers of one-particle density matrices
are exhibited.

\section{Identification of proper condensates}
\setcounter{equation}{0}

\noindent
We consider arbitrary states on the algebra of observable $\fA(\RR^s)$, 
\ie \ positive, linear and normalized
functional $\omega: \fA(\RR^s) \rightarrow \CC$.
Let us recall that any such state gives rise
by the GNS-construction to a representation, where the
state is represented by some unit vector in a Hilbert space
and the elements of the algebra by concrete bounded operators
acting on this space \cite{Ha}. The restrictions of the functional
to the local algebras, $\omega \upharpoonright \fA(\bO)$,
sometimes called partial states, contain the information which
one obtains about $\omega$ by observations
in a given region $\bO \subset \RR^s$.

\medskip
Our criterion, characterizing states containing a proper
condensate, deals with primary states. These
are states where the weak closures of the algebra of observables 
in the respective GNS-representations have a trivial
center. In applications, interesting examples
are pure states and pure phases in case of thermal systems.

\medskip \noindent
\textbf{Definition:} \ Let $\bO \subset \RR^s$ be any
(bounded or unbounded) region. A primary state $\omega$ on
$\fA(\RR^s)$ contains a \textit{proper condensate} in $\bO$ if there 
exists some function $f \in L^2(\bO)$ such that
\be \label{e.1} 
\omega((\mu 1 + a^*(f) a(f))^{-1}) = 0 \, , \quad \mu > 0 \, .
\ee

\medskip \noindent
\textbf{Remark:} \
All $C_0$-functions, \ie \ continuous functions
vanishing at infinity, of $a^*(f) a(f)$ then vanish in the
GNS representation induced by $\omega$, so the single particle
state $f$ is infinitely occupied in all states of the
corresponding GNS-representation; cf.\ also the subsequent
proposition.

\medskip
Mixed states, 
having primary components in their central decomposition   
that contain proper condensates, can be characterized as follows. 
\begin{proposition} \label{p.2.1}
  Let $\omega$ be a state on $\fA(\RR^s)$. 
  Its central decomposition contains a non-negligible set
  of primary states with a proper condensate in $\bO$ 
  iff there is some $f \in L^2(\bO)$ such that
     \be \label{e.2}
   \limsup_{\mu \rightarrow \infty}  \, 
   \mu \ \omega((\mu 1 + a^*(f) a(f))^{-1}) < 1 \, .
  \ee
  There is a central projection $Z$ in the 
  weak closure of $\fA(\RR^s)$ in the corresponding
  GNS-represen\-tation such that $\omega(Z)$ 
  indicates the fraction of proper condensate
  in~$\bO$.
\end{proposition}  
\begin{proof}
  We choose $\mu = \varepsilon^{-1}$, $\varepsilon > 0$,
  and put $A_\varepsilon \doteq (1 + \varepsilon a^*(f) a(f) )^{-1}$. 
  In order to exhibit the properties of these operators in the
  limit of small $\varepsilon$, it is convenient to proceed to
  the (faithful) representation of the full resolvent
  algebra $\fR(\RR^s)$ on Fock space \cite{BuGr}.
  It is generated by the resolvents
  $R(\lambda,g) \doteq (i \lambda 1 + \phi(g))^{-1}$,
  where $\phi(g) = 2^{-1/2}(a^*(g) + a(g))$, 
  $g \in L^2(\RR^s)$, and $\lambda \in \RR \backslash \{ 0 \}$. 

\medskip   
It follows by a straightforward
  computation
  that the sequence $\varepsilon \mapsto A_\varepsilon$ 
  commutes in norm with all
  elements of $\fR(\RR^s)$ in the limit
  of small $\varepsilon$, \ie \ it is a central sequence.
  We briefly sketch the argument. Given 
  $g \in L^2(\RR^s)$ and $\lambda \in \RR \backslash \{ 0 \}$, one has 
  \be \label{e.3}
     [R(\lambda, g), A_\varepsilon] =
  R(\lambda, g) \, [\phi(g), A_\varepsilon] \, R(\lambda, g)
  = R(\lambda, g) \, [\phi(P_f g), A_\varepsilon] \, R(\lambda, g) \, ,
  \ee
  where $P_f$ is the projection onto the ray of $f$. 
  The intertwining relations between $a^*(f)$, $a(f)$, and functions
  of $a^*(f) a(f)$ imply 
  $\| [\phi(f), A_\varepsilon ] \| \leq
  (2 \varepsilon)^{1/2}$. Thus one arrives at the
  bound
\be
  \| [R(\lambda, g), A_\varepsilon] \| \leq
  (2 \varepsilon)^{1/2} \lambda^{-2} \, \| g \| \, ,
 \ee  
  from which the assertion follows. 

  \medskip
  Turning to the GNS representation induced by
  the given state $\omega$, the
  (for decreasing $\varepsilon$) monotonically increasing
  sequence $\varepsilon \mapsto A_\varepsilon \in \fA(\RR^s)$
  converges in the limit of small $\varepsilon$
  in the strong operator topology
  to some operator $P$ in the weak closure
  of $\fA(\RR^s)$. According to the preceding step, 
  it is an element of its center. Moreover,
  putting $A_\varepsilon(\mu) \doteq (\mu 1 + \varepsilon a^*(f)a(f))^{-1}$,
  one obtains
  $\lim_{\varepsilon \searrow 0} \, A_\epsilon(\mu)
  = (1/\mu) \, P$ for any $\mu > 0$.
  It then follows from the
  strong operator convergence of these resolvents and the resolvent equality
  $A_\varepsilon(\mu) - A_\varepsilon(\nu) = (\nu - \mu)  A_\varepsilon(\mu)
  A_\varepsilon(\nu)$ that $P$ is a projection. 
  Thus in a factorial representation it is equal 
  to either $0$ or $1$. In view of
  assumption \eqref{e.1}, the central decomposition of $\omega$
  contains a non-negligible set of
  primary components in which this projection is
  equal to $0$. The central projection $Z = (1 - P)$
  indicates the fraction of primary states in the
  decomposition of $\omega$, where this
  occurs. 
  Since $0 \leq A_{\varepsilon_1} \leq A_{\varepsilon_2}$
  if $\varepsilon_1 \geq \varepsilon_2$, it follows
  that all resolvents $A_\varepsilon$, $\epsilon > 0$, also vanish
  in these states.
  But then all $C_0$-functions of $a^*(f) a(f)$ vanish
  according to standard arguments (Stone-Weierstra{\ss} approximation).
  The converse  statement is trivial. 
\end{proof}
It is important to notice that the preceding characterization of
states exhibiting a proper condensate in some region does not yet
fix the wave functions of the particles forming it. As a
matter of fact, if $h \in L^2(\RR^s)$ is orthogonal to the
function $f$ in the  proposition and $\omega(a^*(h)a(h)) < \infty$,
one finds by  a straightforward estimate that one obtains for 
the resolvent \mbox{$\mu \mapsto (\mu 1 + a^*(f + h) a(f + h))^{-1}$} 
the same upper bound as for $f$ in relation \eqref{e.2}. So there exists
an abundance of wave functions which are infinitely occupied in the presence
of a proper condensate. This is not surprising since counting the
number of particles with a wave function having some overlap with
the proper condensate, no matter how small, must be expected to lead to an
infinite result. 

\medskip
In view of this situation it is more meaningful to determine the particles 
which do not have any overlap with a proper condensate. 

\medskip \noindent
\textbf{Definition:} \ Let $\omega$ be a state on $\fA(\RR^s)$.
A function $f \in L^2(\RR^s)$ is said to be \textit{regular} if 
\be
\lim_{\mu \rightarrow \infty} \, \mu \, 
\omega((\mu 1 + a^*(f) a(f))^{-1}) = 1 \, .
\ee
The set of all regular functions is denoted by $\cR(\RR^s)$ and the
regular functions having support in some region $\bO \subset \RR^s$
are denoted by $\cR(\bO)$. 

\medskip
The structure of the set of regular functions is clarified in
the subsequent lemma.

\begin{lemma}
  Let $\omega$ be a state on $\fA(\RR^s)$. The corresponding set
  $\cR(\RR^s)$ of regular functions is a complex subspace of
  $L^2(\RR^s)$. 
\end{lemma}  
\begin{proof}
  It is apparent that $\cR(\RR^s)$ is stable under multiplication with
  complex numbers. In order to see that it is stable under taking
  sums, we proceed as in the proof of the preceding
  proposition. Let
  $A_\delta(f) \doteq (1 + \delta a^*(f) a(f) )^{-1}$
  for $f \in \cR(\RR^s)$ and $\delta > 0$. These operators
  are monotonically increasing for decreasing $\delta$
  and converge in the limit of small $\delta$
  to the unit operator $1$ in the strong operator topology
  fixed by the underlying state. We then
  consider the operators
  \begin{align}
  & A_{\varepsilon_1}(f_1) \, \big(1 - A_\varepsilon(f_1 + f_2) \big) \,
  A_{\varepsilon_2}(f_2) \nonumber \\
  & = \varepsilon \, 
  A_{\varepsilon_1}(f_1) \, A_\varepsilon(f_1 + f_2) \, a^*(f_1 + f_2)
  a(f_1 + f_2) \,  A_{\varepsilon_2}(f_2) \, .
  \end{align} 
  Commuting the creation and annihilation operators containing the
  function $f_1$ through the middle term 
  to $A_{\varepsilon_1}(f_1)$ and keeping the operators containing
  $f_2$ next to $A_{\varepsilon_2}(f_2)$, one obtains an upper
  bound on the norm of these operators. It yields
  for fixed $\varepsilon_1 > 0$ and $\varepsilon_2 > 0$
  after a straightforward computation  
  \be
  \lim_{\varepsilon \searrow 0} \, 
  \|  A_{\varepsilon_1}(f_1) \, \big(1 - A_\varepsilon(f_1 + f_2) \big)
  A_{\varepsilon_2} \big) \| =  0 \, . 
  \ee
  Since $\| A_\delta(f) \| \leq 1$, it then follows by a
  three-epsilon argument for $f_1, f_2 \in \cR(\RR^s)$ that
  \be
  \lim_{\varepsilon \searrow 0} \, \omega(1 - A_\varepsilon(f_1 + f_2)) = 0
  \, ,
  \ee
  proving the statement.
\end{proof}   
The appearance of proper condensates was studied in \cite{BaBu} in 
non-interacting theories. For equilibrium states in a fixed
trapping potential they appear in the limit of infinite
particle numbers (maximal chemical potential). In these cases the
regular spaces $\cR(\RR^s)$ are closed subspaces of
$L^2(\RR^s)$. Their \textit{singular} orthogonal complement $\cS(\RR^s)$
describes the wave functions of the particles appearing 
in proper condensates in the limit states.

\medskip 
Proceeding first to the
thermodynamic limit by unfolding the trapping potential
and going subsequently to the limit of maximal chemical
potential, it turns out that the regular spaces $\cR(\RR^s)$
are dense in $L^2(\RR^s)$, but not closed \cite{BaBu}. 
This observation has an easy explanation: in the thermodynamic
limit the constituents of proper condensates can no longer
be described by normalizable wave functions, they become 
improper states, characterized by constant functions, polynomials
\etc. The elements of the subspaces $\cR(\RR^s)$ then serve as
test functions. Regardless of this fact, the restrictions of the improper
states to bounded regions $\bO$ are normalizable. As a consequence,
the local regular spaces $\cR(\bO)$ are closed. In one and two dimensions
they have orthogonal complements $\cS(\bO)$ in $L^2(\bO)$, describing 
proper condensates. In higher dimensions such non-trivial
orthogonal complements appear in the limit of infinite particle
densities. Examples, illustrating these facts, were
presented in~\cite{BaBu}.  

\medskip 
It seems worthwhile to extend this analysis to 
combined approximations, where the approach to the 
thermodynamic limit and the increase of particle
numbers are coupled. In order to illustrate the
present concepts, we perform such an analysis
in the subsequent section for non-interacting
bosons in regular trapping potentials. These potentials are 
unfolded quite arbitrarily in order to occupy an increasing number 
of particles, thereby approaching  
different kinds of thermodynamic limits. It turns out that
by unfolding the trapping potential too tardily, proper
condensates appear if one rapidly increases the particle
number. On the other hand, a swift unfolding of the trapping
potential leads exclusively to regular excitations.
In this manner one can exhibit distinct 
local wave functions $\cS(\bO)$ of proper condensates and
study in suitable approximations the onset of condensation.

\medskip
A major challenge, however, is the treatment of
theories describing interactions. We do not
tackle this demanding issue here and only summarize
the general strategy for the verification of
the appearance of
proper condensates in models, which is suggested by our results. 
Thinking of such applications, let
$\omega_n$ be a sequence of states, which is labeled by the 
particle number $n \in \NN$. In order to
exhibit the relevant structures, one
must proceed to the limit of this
sequence. The sequence may not converge,
but since we are dealing with an
algebra of bounded operators, it always
has limit points according to standard
compactness arguments. As a matter of
fact, one does not need information
about the limits on the full algebra.
It suffices to determine the expectation
values of resolvents of particle
number operators. Their 
analysis consists of the following steps. 

\bigskip \noindent 
\textit{Step 1:}  Determine for the given sequence of states 
in the limit of large $n \in \NN$ the regular space $\cR(\RR^s)$. If it
coincides with $L^2(\RR^s)$, there is no sign for proper 
condensation. If it is closed and has a non-trivial
orthogonal complement $\cS(\RR^s)$, the latter 
space describes the wave functions of a proper condensate.
If  $\cR(\RR^s)$ is dense, but not closed, proceed to the local
regular spaces $\cR(\bO)$. They are expected to be closed.
If they have an orthogonal complement $\cS(\bO)$, 
it describes locally a  proper condensate. 

\medskip \noindent 
\textit{Step 2:} Proceed to the analysis of the
structure of condensates appearing in the states
with a finite particle number. 
In cases of interest, the spaces $\cS(\bO)$
are expected to be finite dimensional. The corresponding
subalgebras of observables $\fA_{\cS}(\bO)$, 
involving resolvents of operators with
functions in $\cS(\bO)$, then have a simple structure.
In particular, the operators $N_{\cS}(\bO)$, determining
the number of condensate particles in $\bO$, are affiliated
with $\fA_{\cS}(\bO)$. Moreover, the restrictions of  $\fA_{\cS}(\bO)$
to states with a finite particle number are isomorphic
to matrix algebras. Thus, for any given total particle number
$n \in \NN$, one can accurately determine and
manipulate the properties of the
condensate fraction, such as its particle content.
This does not affect  properties of the regular
excitations, described by the regular algebra
of observables $\fA_{\cR}(\bO)$, which is assigned to the regular
functions $\cR(\bO)$ and commutes with
$\fA_{\cS}(\bO)$. This feature resembles the
empirical fact that for given material content
of a system, it is frequently possible to split it quite
arbitrarily into different phases.  

\medskip \noindent 
\textit{Step 3:} Study the onset of proper condensation.
According to folklore, there do not exist phase 
transitions in finite systems. But, having identified
the condensate spaces $\cS(\bO)$ appearing in the limit of
infinite particle numbers, one can analyze  the 
emergence of proper condensates in a given sequence of
states. This is accomplished by computing the expectation values of the
number operators $N_{\cS}(\bO)$ for the proper condensate,
$
n \mapsto \omega_n(N_{\cS}(\bO)) \, , 
$
and comparing it with the total number of particles in $\bO$,
given by $n \mapsto \omega_n(N(\bO))$.  
Particularly interesting are the cases where the difference
between these data remains bounded in the limit,
\be
\limsup_n \, \omega_n(N(\bO) - N_{\cS}(\bO)) = m_{\cR}(\bO) < \infty \, .
\ee
This happens either if the total number of particles in $\bO$
stays bounded, or if the regular components, described by
$\cR(\bO)$, are saturated in the limit. One can then define
the onset of proper condensation in $\bO$
by the number $n_c(\bO)$ of particles in the condensate
which surpass this limit,  
\be
\omega_n(N_{\cS}(\bO)) \geq m_{\cR}(\bO) \, , \quad n \geq n_c(\bO) \, .
\ee
In equilibrium states one may then use this
condition in order to study its relation with other data, such
as the temperature of the states. 

\vspace*{-4mm}
\section{Application of the framework to   
  non-interacting ground states}
\setcounter{equation}{0}  

\medskip \noindent 
In order to illustrate the concept of proper condensation and its 
usage in the interpretation of a theory, we consider the simple example
of non-interacting, trapped bosons in the
ground state and passages to the thermodynamic limit.
As already mentioned, there exist various approximations of this limit
which differ by the spaces of wave functions, describing the resulting 
proper condensates. Depending on the approximation,
one can exhibit in states
with a finite particle number
more or less detailed spatial patterns of the
proper condensates, establish the existence of coexisting phases, and encounter
critical densities of the regular fractions.  

\medskip 
   To keep the discussion simple, we assume that the
   one-particle Hamiltonian on the space of wave functions $L^2(\RR^s)$ 
   is of the form $H_1 = \bP^2 + V(\bQ)$, where $\bP, \bQ$ are the
   momentum and position operators and $V$ is some
   real analytic trapping potential, tending
   sufficiently rapidly to infinity at large
   distances. (As a matter of fact, less stringent
   smoothness properties of the potential would be 
   sufficient.) The normalized ground state wave function
   $\bx \mapsto g_1(\bx) \in L^2(\RR^s)$ of the single particle
   states is then
   also real analytic \cite[Ch.\ 3.4]{BeSh}. The ground state of the 
   corresponding
   non-interacting system of $n$ particles in the  trapping potential
   is given by
   \be
   \Omega_n = (n!)^{-1/2} \, a^*(g_1)^n \, \Omega_0 \, ,
   \ee
   where $a^*$ denotes the canonical creation
   operators and $\Omega_0$ the Fock vacuum vector.
   
\medskip
   We will proceed to the
   thermodynamic limit by unfolding the trapping potential, 
   which yields for $\lambda > 0$ the scaled Hamiltonians
   $H_\lambda \doteq \bP^2 + \lambda^2 \, V(\lambda \bQ)$.
   The corresponding scaled and normalized ground states are 
   $\bx \mapsto g_\lambda(\bx) \doteq \lambda^{s/2} g_1(\lambda \bx)$ 
   and the scaled $n$-particle states are
   \be \label{be.2} 
   \Omega_{n , \lambda} = (n!)^{-1/2} \, a^*(g_\lambda) \, \Omega_0 \, ,
   \quad n \in \NN \, , \ \lambda > 0 \, .
   \ee 
   In the passage to the thermodynamic limit we need to
   couple the particle number $n$ and the scaling parameter
   $\lambda$,
   $n \mapsto \lambda(n)$, which will be
   accomplished in different ways.

   \subsection{Regular and singular wave functions} 
\noindent We determine now the spaces of regular wave functions
$\cR(\RR^s)$, describing finitely occupied
states, and of singular wave functions
$\cS(\RR^s)$, describing proper condensates in
the thermodynamic limit. They are identified by
computing in a given sequence of states the expectation values
of resolvents of the particle number operators
$a^*(f)a(f)$, \mbox{$f \in L^2(\RR^s)$},  and proceeding 
to the limit of infinite particle numbers. It 
follows from the subsequent lemma that
the relevant information is encoded in the 
transition probabilities between $f$ and the ground states,
multiplied by the particle number,
\viz \ $n \, | \langle g_{\lambda(n)}, f \rangle |^2$. 
\begin{lemma} \label{l.3.1}
     Let $g_n \in L^2(\RR^s)$ 
     be a sequence of normalized wave functions and let 
     $\omega_n$ be the $n$-particle states, given by the vectors 
     $\Omega_n \doteq (n!)^{-1/2} a^{*}(g_n) ^n \, \Omega_0$, \ 
     $n \in \NN$. If $f \in  L^2(\RR^s)$ is 
     such that $\lim_n n \, | \langle g_n, f \rangle |^2 = 0 $, 
     one obtains for all $\mu > 0$ 
     \be
     \lim_n \, \omega_n((\mu 1 + a^*(f)a(f))^{-1}) = \mu^{-1} \, .
     \ee
     If $f \in  L^2(\RR^s)$ satisfies 
     $\lim_n n \, | \langle g_n, f \rangle |^2 = \infty $,
     then, for all  $\mu > 0$,
     \be
     \lim_n \,
     \omega_n((\mu 1 + a^*(f)a(f))^{-1}) = 0 \, . 
     \ee
   \end{lemma}
   \noindent \textbf{Remark:} In the former case, the wave function $f$
   is regular and the limit state coincides on the corresponding observables
   with the Fock vacuum. In the latter case, the wave function $f$ 
   indicates the presence of some proper condensate.
   \begin{proof}
  Without loss of generality we may assume that $f$ is normalized,
  hence $|\langle g_n, f \rangle | \leq 1$. Putting 
  $R_f(\mu) \doteq (\mu 1 + a^*(f) a(f))^{-1}$ and applying 
  arguments given in Proposition \ref{p.2.1}, we have 
  \begin{align}  
  &  \omega_n(R_f(\mu)) = \langle \Omega_n, \, R_f(\mu) \, \Omega_n \rangle
    \nonumber  \\
    & = (n)^{-1/2} \langle \Omega_{n-1} [a(P_fg_n), R_f(\mu))] \, \Omega_n \rangle
    + (n)^{-1/2}  \langle \Omega_{n-1},  R_f(\mu)) \, a(g_n) \, \Omega_n \rangle
    \nonumber \\
  & = | \langle g_n, f \rangle |^2 \, \omega_{n-1}(R_f(\mu + 1) - R_f(\mu)) 
  +  \omega_{n-1}(R_f(\mu)) \nonumber \\
  & =  | \langle g_n, f \rangle |^2 \, \omega_{n-1}(R_f(\mu + 1)) +
  (1 - | \langle g_n, f \rangle) |^2 \,  \omega_{n-1}(R_f(\mu)) \, . 
  \end{align}
  Making use of 
  $\omega_0(R_f(\mu)) = 1/\mu$, it follows by induction
  that
  \be \label{be.13}
  \omega_n(R_f(\mu)) = \sum_{k = 0}^n  (\mu + k)^{-1}
  \mbox{\large $\binom{n}{k}$} \,  | \langle g_n, f \rangle |^{2k} \,
  (1 -  | \langle g_n, f \rangle |^2)^{n-k} \, .
  \ee
  Now let $\lim_n n \, |\langle g_n, f \rangle|^2 = 0$. Then
  \begin{align} \label{be.12}
  &  |   \omega_n(R_f(\mu)) - \mu^{-1}  |
    =  \sum_{k = 1}^n \big(\mu^{-1} -  (\mu + k)^{-1} \big)
  \mbox{\large $\binom{n}{k}$} \,  | \langle g_n, f \rangle |^{2k} \,
  (1 -  | \langle g_n, f \rangle |^2)^{n-k}  \nonumber \\
  & = n \, | \langle g_n, f \rangle |^2 \,
  \sum_{l = 0}^{n-1} \mu^{-1} (\mu + l + 1)^{-1}
   \mbox{\large $\binom{n -1}{l}$} \,  | \langle g_n, f \rangle |^{2l} \,
   (1 -  | \langle g_n, f \rangle |^2)^{(n -1) -l} \nonumber \\
  & \leq  n \, | \langle g_n, f \rangle |^2 \,
   \mu^{-2} \sum_{l = 0}^{n-1} 
   \mbox{\large $\binom{n -1}{l}$} \,  | \langle g_n, f \rangle |^{2l} \,
   (1 -  | \langle g_n, f \rangle |^2)^{(n -1) -l} \nonumber \\ 
 &  =  n \, | \langle g_n, f \rangle |^2 \, \mu^{-2} \, .
  \end{align}  
  Hence the sequence  converges to $0$, as stated.
  
\medskip 
  Next, we consider the case that
  $\lim_n \, n \,  | \langle g_n, f \rangle |^2  = \infty$. 
  Thus we may assume that
  $ | \langle g_n, f \rangle |^2 \neq 0$ and obtain  
  \begin{align} \label{be.14}
  &  \omega_n(R_f(\mu))
    =  \sum_{k = 0}^n (\mu + k)^{-1} 
  \mbox{\large $\binom{n}{k}$} \,  | \langle g_n, f \rangle |^{2k} \,
  (1 -  | \langle g_n, f \rangle |^2)^{n-k}  \nonumber \\
  & = \big( (n+1) \, | \langle g_n, f \rangle |^2 \big)^{-1} \,
  \sum_{l = 1}^{n+1} (\mu - 1 + l)^{-1} \, l \,  
  \mbox{\large $\binom{n + 1}{l}$} \,
  | \langle g_n, f \rangle |^{2 l} \,
   (1 -  | \langle g_n, f \rangle |^2)^{(n +1) -l} \nonumber \\
  & \leq 
 \big( (n+1) \, | \langle g_n, f \rangle |^2 \big)^{-1} 
  (1 + \mu^{-1}) \, \sum_{l = 0}^{n+1} \,   \mbox{\large $\binom{n + 1}{l}$}
   | \langle g_n, f \rangle |^{2l} 
   \big(1 -  | \langle g_n, f \rangle |^2 \big)^{(n +1) -l} \nonumber \\
& =  \big( (n+1) \, | \langle g_n, f \rangle |^2 \big)^{-1} \,  (1 + \mu^{-1})
\, .
  \end{align}
The expression in the last line converges to $0$, 
completing the proof. 
 \end{proof}  

The wave functions that are regular 
in the thermodynamic limit can 
now be determined. Let us recall that it is
meaningful to restrict 
attention in this limit to regular functions with compact
support. 
Because, the proper condensates are globally described by distributions,
but their restrictions to bounded regions $\bO$ are 
square integrable functions $\cS(\bO)$, which lie in the orthogonal complements
of the regular functions $\cR(\bO)$. 
According to the preceding lemma, 
the transition probabilities of
the wave functions $f$ to the scaled ground states,
multiplied by the particle number, contain the relevant
information about their interpretation as members of
either one of these spaces. 

\medskip 
In order to see how this assignment is related to properties of the
ground state wave functions, let
$\lambda \mapsto n(\lambda) \doteq c  \lambda^{- \kappa}$ for some 
$\kappa > 0$. There exists a corresponding dense set of functions
$f \in L^2(\RR^s)$ such that
$\lambda^{-\kappa} \, |\langle g_\lambda, f \rangle|^2 \rightarrow 0$
in the scaling limit $\lambda \searrow 0$. For the proof
we make use of the fact that
   $g_1$, being real analytic, can be expanded in a Taylor series
   about~$0$. Thus, for any given $k \in \NN$, there exists a
   polynomial $\bx \mapsto P_k(\bx)$ of degree $(k-1)$ such that  
   $|g_1(\bx) - P_k(\bx)| \leq c_k \, | \bx |^k$ for
   $|\bx| < R_0$, where $R_0$ depends on the analyticity
   properties of $g_1$. 
   This yields for the integrals of the scaled functions
   over the balls
   $B_R = \{ \bx \in \RR^s : |\bx| < R \}$ and
   scalings $0 < \lambda < R_0 / R$
   \be
   \int_{B_R} \! d\bx \, \lambda^s \,  |g_1(\lambda \bx) - P_k(\lambda \bx)|^2
   \leq c_k(R) \, \lambda^{s + 2k} \, .
   \ee

   Now let $f \in L^2(\RR^s)$
   be any function with compact support such that its
   (entire analytic) Fourier transform $\widetilde{f}$ vanishes
   sufficiently rapidly at the origin. More precisely,
   $P_k(i \lambda \, \bpartial_{\bp}) \,  \widetilde{f}(\bp) |_{\bp = 0} = 0$
   for all $\lambda > 0$, where $\bpartial_{\bp}$ denotes the
   gradient in momentum space. Since the linear span
   of the scaled polynomials is finite dimensional, such
   functions exist. As a matter of fact, their linear span is
   dense in~$L^2(\RR^s)$. Now, given any such function, it has
   support in $B_R$ for sufficiently large~$R$, so one
   obtains for small $\lambda$  
   \be
   | \langle g_\lambda, f \rangle |^2
   =   | \langle ( g_\lambda - P_{k , \lambda} ) , f \rangle |^2
   \leq  {c}_k(R) \, \| f \|^2  \, \lambda^{s + 2k} \, .
   \ee 
   Let $k \in \NN_0$ be  
   the smallest number such that 
   \mbox{$k > (\kappa - s)/2$}. It follows that 
   $\lambda^{-\kappa} \, |\langle g_\lambda, f \rangle|^2 \rightarrow 0$
   in the limit of small $\lambda$, as stated.

   \medskip 
   If $\kappa < s$,
   the above limit is equal to $0$ for all functions $f$ with compact support.
   So the local spaces of regular functions $\cR(\bO)$
   do not have an orthogonal
   complement in $L^2(\bO)$ for any bounded region $\bO \subset \RR^s$, 
   there appear no proper condensates in the thermodynamic limit. 
   If $\kappa > s$ there arise for any $f$ with compact support 
   the following clear-cut alternatives: either the limit is $0$, 
   or it approaches $\infty$. In the former case, $f$ is a regular
   member of some space $\cR(\bO)$. It is orthogonal to
   the homogeneous parts of the approximating polynomial $P_k$, describing
   a proper condensate, whose square integrable restrictions
   to bounded regions $\bO$ form the singular spaces $\cS(\bO)$. 
   In the latter case, the scalar product of $f$ with a condensate
   wave function is different from $0$. It then follows that
   $\lambda^{-\kappa} \, |\langle g_\lambda, f \rangle|^2$
   approaches infinity for small $\lambda$. Note that $f$
   may not be a member of a space $\cS(\bO)$, it must merely 
   have some overlap with a singular function. 
   In spite of this feature, one can unambiguously identify the regular and
   singular wave functions for the given values of $\kappa$.

\subsection{Coexistence of phases}
   \noindent 
   Whereas in the preceding two cases either the proper condensate or
   the regular excitations dominate in the thermodynamic limit, the situation
   is different for $\kappa = s$. It turns out in this
   intermediate case that the regular
   excitations can coexist with a condensate in arbitrary portions,
   \viz  they form coexisting phases. 
  So let $n \mapsto \lambda(n) \doteq \sigma \, n^{-1/s}$ be
  given, $\sigma > 0$. The resulting states 
  are locally normal in the thermodynamic limit, \ie \ their restrictions
  to any local observable algebra $\fA(\bO)$ can be represented by a density
  matrix in Fock space, depending on the bounded region $\bO$. There
  are no proper condensates in these states; but they appear
  if one lets $\sigma$ tend to infinity. So, again, one can 
  identify the regular wave functions
  and the wave functions of the asymptotic proper condensate  
  and disentangle the contributions which are due to 
  the onset of proper condensation from those of the regular excitations.

  The strategy of proof is the same as in the preceding lemma. 
  Assuming that the unscaled
  ground state wave function $g_1$ is different from $0$ at the
  origin, one obtains by a straightforward computation 
  for functions $f \in L^2(\RR^s)$ with compact support 
  \be
  \nu_f  \doteq \lim_n \,
  n \, | \langle g_{\lambda(n)} , f \rangle |^2 
  = \sigma^s \, |g_1(0)|^2 \, \Big| \! 
  \int \! d \bx \, f(\bx)  \Big|^2 \, .
  \ee 
  Thus if $\sigma$ tends to infinity, this expression
  diverges, unless the Fourier transform of $f$ vanishes
  at $0$. So the spaces of regular
  functions $\cR(\bO)$ consist of 
  square integrable functions with support in $\bO$
  and Fourier transforms which vanish
  at the origin. The condensate spaces $\cS(\bO)$ 
  consist of their orthogonal complements in $L^2(\bO)$, 
  \ie \ the constant functions in~$\bO$. Making use of 
  arguments in the proof of the preceding lemma, we
  obtain the following result. 
  \begin{lemma}
    Let  $\sigma > 0$, let $n \mapsto \lambda(n) = \sigma \, n^{- 1/s}$,
    and let $\omega_{n, \lambda(n)}$ be the states 
    deter\-mined by the  vectors \ 
    $\Omega_{n, \lambda(n)}$ in equation \eqref{be.2}, $n \in \NN$. 
    For any compactly sup\-ported function
    $f \in L^2(\RR^s)$ and $\mu > 0$, one has 
    \be
    \lim_n \, \omega_{n, \lambda(n)}((\mu 1 + a^*(f)a(f))^{-1}) =
    \mu^{-1} \Big(1 - \nu_f \, e^{- \nu_f}
    \int_0^1 \! d \upsilon \, \upsilon^\mu \, e^{\, \nu_f \, \upsilon}
    \Big) \, .
    \ee
  \end{lemma}  
  \begin{proof}
    Making use of relation \eqref{be.13} and putting
    $\nu_{f,n} \doteq n \,  | \langle g_{\lambda(n)} , f \rangle |^2$,
    one obtains  
    \begin{align} 
      & \omega_{n, \lambda(n)}(R_f(\mu)) \nonumber \\
      & = (1 - \nu_{f,n}/n )^{\, n} \, 
    \sum_{k = 0}^n  (\mu + k)^{-1} (k!)^{-1} \nu_{f,n}^k \ 
    \prod_{l = 0}^{k-1} \big( (1 - l/n) \, (1- \nu_{f,n}/n)^{-1} \big)      \, .
    \end{align} 
    In the limit of large $n$ one has
    $\nu_{f,n} \rightarrow \nu_f$ and
    $(1 - \nu_{f,n}/n )^{\, n} \rightarrow e^{- \nu_f}$.
    The product under the sum has the upper bound
    $e^{\nu_{f,n}}$ and hence is uniformly bounded
    in $k$ and $n$. Thus, by an application of the dominated
    convergence theorem, one arrives~at
      \be \label{be.17}
      \lim_{n \rightarrow \infty} \omega_{n, \lambda(n)} (R_f(\mu)) =
      e^{- \nu_f} \, \sum_{k = 0}^\infty \, 
      (\mu + k)^{-1} (k!)^{-1} \, \nu_f^k \, .
      \ee
      This result coincides with the expression given
      in the statement, as can be established by a routine computation. 
  \end{proof}   
  It follows from this lemma that 
  the resolvents of the particle number operators 
  $N(f) \doteq a^*(f) a(f)$ for compactly supported,
  normalized functions
  $f$ are regular in all limit states. This is
  in accord with the statement that these states
  are locally normal. It also follows from  the lemma
  that the mean of $N(f)$ in the limit states is given by
  $\omega_\infty(N(f)) = \nu_f$. Thus it is $0$ for regular
  functions in $\cR(\bO)$, which indicate the Fock vacuum. 
  For normalized singular functions $s \in \cS(\bO)$ one obtains
  $\omega_\infty(N(s)) = \nu_s = \sigma^s |g_1(0)|^2 \, |\bO|$,
  which describes a homogeneous condensate with a density
  determined by $\sigma$. As already mentioned, it
  becomes a proper condensate if $\sigma$ tends to infinity.

\subsection{Spatial structure of condensates}
\noindent 
In case of the special scalings considered in the preceding subsection, 
the condensates which appeared there were spatially homogeneous.
Yet it is an empirical fact that systems with given
material content can often be split quite arbitrarily into
different phases localized in differing regions,
so the corresponding states are inhomogeneous. 

\medskip 
We will show now that such more complex structures emerge
in the present model 
if one proceeds to scalings $n \mapsto \lambda(n) = \sigma \, n^{- 1/\kappa}$
with $\kappa > s$. As we have seen, the proper condensates
are determined by polynomials $\bx \mapsto P_k(\bx)$,
approximating the wave function of the ground state, where
\mbox{$k > (\kappa - s)/2$} denotes the degree of approximation. 
Because of the scaling involved in the passage to the
thermodynamic limit, the homogeneous pieces of this
polynomial are of interest. Their restrictions to
bounded regions $\bO$ form the finite dimensional spaces
$\cS(\bO)$ of singular wave functions, having the
regular functions $\cR(\bO)$ in their orthogonal
complement. 

\medskip
We keep the bounded region $\bO$ fixed in the following
and split the $n$-particle space, formed by wave functions
with support in this region, into a sum of symmetric tensor products
containing regular, respectively singular functions,
 \be
 \cF_n(L^2(\bO)) = \sum_{l=0}^n \cF_l(\cS(\bO)) \otimes \cF_{n-l}(\cR(\bO)) \, ,
 \ee
 in an obvious notation. 
 The action of the algebra of condensate observables $\fA(\cS(\bO))$, being 
 generated by resolvents with functions in $\cS(\bO)$,
 affects only the finite dimensional subspaces
 $\cF_l(\cS(\bO))$, the regular
 subspaces $\cF_{n-l}(\cR(\bO))$ remain untouched.
 Thus  $\fA(\cS(\bO))$ acts like a matrix algebra
 on $\cF_n(L^2(\bO))$. Whence, by a finite number of operations,  one can
 analyze and modify the condensate part of the
 $n$-particle states. In particular, one can discriminate 
 the wave functions of particles in the condensate.
 Doing this in all regions 
 $\bO$, one can thereby unravel their spatial structures.

\medskip 
 Furthermore, 
 given any number $0 \leq k \leq n$ of particles in the condensate,
 there is a projection $P_k \in \fA(\cS(\bO))$ which
 projects onto the subspace, 
 $P_k : \cF_n(L^2(\bO)) \rightarrow \cF_k(\cS(\bO)) \otimes
 \cF_{n-k}(\cR(\bO))$.
 Applying this projection to the vector $\Omega_{n, \lambda(n)}$ in 
 equation \eqref{be.2}, one obtains, apart from some numerical factor,  
 \be
 P_k \, \Omega_{n, \lambda(n)} \approx
 a^*(P_1 \, g_{\lambda(n)})^k \, a^*((1- P_1) \,
 g_{\lambda(n)})^{n-k}  \, \Omega_0 \, .
 \ee
 These vectors define  product states $\omega_{n,k}$ on the algebra
 $\fA(\cS(\bO)) \otimes \fA(\cR(\bO))$, 
 \be
 \omega_{n,k}(A B) = \overline{\omega}_k(A)
 \, \overline{\omega}_{n-k}(B) \, , \quad
 A \in \fA(\cS(\bO)), \ B \in  \fA(\cR(\bO)) \, .
 \ee
 Here $\overline{\omega}_l$ are the $l$-particle states
 with state vectors
 $\overline{\Omega}_{l,\lambda(n)} \doteq (l!)^{-1/2}
 a^*(g_{\lambda(n)})^l \, \Omega_0$, $0 \leq l \leq n$.
 Thus, picking any $0 \leq k \leq n$, 
 the fraction $k/n$ of condensate in the
 resulting states $\omega_{n,k}$ 
 can be fixed without affecting the total number $n$ of particles.  

\medskip 
So, to summarize, the amount of condensate in the given finite system
at zero temperature can be arbitrarily adjusted, just like its
spatial patterns, which are determined by the choice of the 
trapping potential.
 
\subsection{Critical densities}

\noindent
In case of scalings for which proper condensates appear in 
the thermodynamic limit, these outrun eventually the occupation numbers
of all other states. Given a space of singular
wave functions $\cS(\bO)$, which has been identified for
some particular scaling, it is therefore of interest to
analyze the fate of the corresponding regular excitations
by changing the scaling and studying the number of particles with
wave functions in $\cR(\bO) \doteq \cS(\bO)^\perp \cap L^2(\bO)$.
These numbers can easily be determined.

\medskip
Let $N(\bO)$, $N_{\cS}(\bO)$, and
$N_{\cR}(\bO)$ be, for given region $\bO$, the total number
operator, the
operator counting the particles with wave functions in $\cS(\bO)$,
respectively in the corresponding regular space $\cR(\bO)$. 
In order to simplify the discussion, let us assume
that $\cS(\bO)$ is one-dimensional, consisting of multiples 
of the normalized characteristic function $\chi_{\bO}$
of $\bO$. We also note that, without restriction of
generality, we may assume that the ground state
wave function $g_1$ is positive.
Putting $\lambda(n) = \lambda$ for a moment 
and keeping $n$ fixed, we have 
\begin{align}
  \omega_{n, \lambda}(N_{\cR}(\bO)) & =
  \omega_{n, \lambda}(N(\bO)) - \omega_{n, \lambda}(N_{\cS}(\bO)) \nonumber \\
& = n \lambda^s \, \Big(  \int_{\bO} \! d\bx \, g_1(\lambda \bx)^2
- 
\Big( \int_{\bO} \! d\bx \, g_1(\lambda \bx) \, \chi_{\bO}(\bx)
\Big)^2 \, \Big)   \, . 
\end{align}
The expression in the second line is non-negative,
analytic for sufficiently small $\lambda$, and it is different from $0$.
It is also apparent that the two terms in the bracket
cancel each other at $\lambda = 0$. Thus there is
some number $l \in \NN$ and some constant
$c(\bO)  >  0$ such that
for small $\lambda$ one has in leading order 
$ \omega_{n, \lambda}(N_{\cR}(\bO)) \approx c(\bO) \, n \lambda^{s + l}$. 

\medskip  
It follows that for scalings
$n \mapsto \lambda(n) \approx n^{-1/\kappa}$ with $\kappa < s + l$ 
all excitations in $\cR(\bO)$ disappear in the thermodynamic 
limit, only the states in
$\cS(\bO)$ are 
occupied, eventually. If $\kappa  > s + l$, the number of excitations in
$\cR(\bO)$ tends to infinity, indicating the
appearance of further modes which contribute to the
proper condensate. Of particular interest are the 
scalings $n \mapsto \lambda(n) = \sigma \, n^{-1/(s + l)}$. 
There one arrives at
\be
\lim_n \, \omega_{n, \lambda(n)} (N_{\cR}(\bO)) = c(\bO) \, \sigma^{s + l} \, .
\ee
Thus, for these scalings, the particles with wave functions in $\cR(\bO)$
have a critical local density. So the expectation values of 
the number of condensate particles 
$n \mapsto \omega_{n, \lambda(n)}(N_{\cS}(\bO))$ with wave functions
in $\cS(\bO)$ become dominant if this value is reached,
akin to the appearance of a phase transition. 
So the points illustrated here
show that the concept of proper condensates leads to
a refined understanding of the onset of condensation. 

\vspace*{-4mm}
\section{Macroscopic occupation and proper condensation}
\setcounter{equation}{0}  

\noindent
We compare now our notion of proper condensates 
with the concept of macroscopic occupation numbers, 
proposed by Onsager and Penrose \cite{OnPe}, which is 
frequently used in the literature. It also deals with
sequences of $n$-particle states $\omega_{n}$, where one determines
for each $n$ the dominant occupation number of single
particle states, given by the maximal 
eigenvalue of the corresponding 
one-particle density matrices, $n \in \NN$.
In order to distinguish this concept
from our notions, we introduce the following terminology.

\medskip \noindent
\textbf{Definition:} \ Let $\omega_{n}$ be a sequence of
$n$-particle states, $n \in \NN$. This sequence describes 
\textit{growing condensates} if there exist normalized functions  
$f_n \in L^2(\RR^s)$, $n \in \NN$, and some constant $0 < \delta \leq 1$ 
such that
\be \label{e.10}
\limsup_n \, (1/n) \, \omega_{n}(a^*(f_n)a(f_n)) \geq \delta \, .
\ee
According to standard terminology, the corresponding single
particle states are macroscopically occupied in this case. 

\medskip
Being based in a clear-cut manner on the convenient notion
of one-particle density
matrices, this characterization of growing condensates
has found numerous applications
in the analysis of models. Yet, in spite of these
successes, it is not fully
satisfactory in some respects. First, since the
functions $f_n$ vary with~$n$, it does not give an
answer to the question whether the condensates can
be described in the limit of large $n$
by specific wave functions or, more generally, specific improper states. 
Second, relation \eqref{e.10} does not characterize by itself the
condensate functions $f_n$. As a matter of fact, adding to
$f_n$ any function $h$ for which
$\limsup_n \omega_n(a^*(h) a(h)) < \infty$, 
the resulting functions still comply with this
relation. Since it will not always be possible to
determine precisely the eigenfunctions of the
one-particle density matrices, the structure of the
condensates in the limit remains to be even more
obscure. Third, one might ask whether
the dependence of the expectation values
in relation \eqref{e.10}
on the particle number $n$ is really crucial. Taking into account that
the universe contains about $10^{80}$ atoms, it 
may still seem meaningful to speak of macroscopic occupation if one
replaces the prefactor $(1/n)$ by
$(\ln(n)/n)$, say. It turns out, however, that one then opens
Pandora's box. As is shown in the appendix, the
set of functions $f$, satisfying such mildly weaker
conditions in a sequence of states exhibiting
growing condensates, is huge (it is of second category). In a rough
analogy: it is as big as the set of non-rational real numbers
compared to the number of rationals. 

\medskip
In view of the results obtained in the preceding sections, 
it is apparent what is missing in order to solve these conceptual 
problems: one must determine the wave functions which are
regular in the limit. Restricting our attention to the
one-particle density matrices, we are led to the following
definition.

\medskip \noindent 
\textbf{Definition:} Let $\omega_n$ be a sequence of $n$-particle
states, $n \in \NN$. A function \mbox{$f \in L^2(\RR^s)$} is said to
be one-particle regular if
\be
\limsup_n \, \omega_n(a^*(f) a(f)) < \infty \, .
\ee
The complex subspace of one-particle regular functions is denoted
by $\underline{\cR}(\RR^s)$.  

\medskip
Since the number operators, counting particles with given wave
function, are unbounded, the space of one-particle regular functions
is in general smaller than the full space of regular functions.
This is shown in the
subsequent lemma. What is of more interest is the observation that if
the space $\underline{\cR}(\RR^s)$ is closed and has a
finite dimensional orthogonal complement, there
exist functions $f$, not depending on the particle
number, for which condition \eqref{e.10} is satisfied. 
\begin{lemma}
  Let $\omega_n$ be a sequence of $n$-particle states, $n \in \NN$,
  and let $\underline{\cR}(\RR^s)$ be the corresponding one-particle
  regular functions.
  \begin{itemize}
  \item[(i)]  $\underline{\cR}(\RR^s) \subset \cR(\RR^s)$, where
    ${\cR}(\RR^s)$ is the space of regular functions, 
    defined in the Section 2. 
  \item[(ii)]  Let the sequence of states exhibit growing
    condensates and let the corresponding space $\underline{\cR}(\RR^s)$
    be closed and have a finite dimensional orthogonal complement. There is
    a function $f$ in this complement, which does not
    depend on the particle number, for which relation
    \eqref{e.10} is satisfied with some lower bound
    $\delta^\prime > 0$.
  \end{itemize}  
\end{lemma}  
\begin{proof}
  (i)  The first statement is a  consequence of the simple estimate
  \begin{align} 
  \omega_n\big( (1 - (1 + \varepsilon a^*(f) a(f))^{-1})\big) & =
  \varepsilon \,
  \omega_n\big( a^*(f)a(f) (1 + \varepsilon a^*(f) a(f))^{-1} \big)
    \nonumber \\ 
  & \leq  \varepsilon \, \omega_n \big( a^*(f)a(f) \big) \, .
  \end{align}
  Thus if $\omega_\infty$ is any limit point of the given sequence of states,
  one has
  \be
  \omega_\infty\big( (1 - (1 + \varepsilon a^*(f) a(f))^{-1})\big)
  \leq \varepsilon \, \limsup_n \omega_n\big( a^*(f) a(f) \big) 
= \varepsilon \, c_f \, , \
  f \in \underline{\cR}(\RR^s) \, . 
  \ee
  So the resolvents of the particle number operators
  assigned to functions \mbox{$f \in \underline{\cR}(\RR^s)$} converge
  in the limit states
  to the unit operator $1$ in the limit of small $\varepsilon$.

  \medskip \noindent 
  (ii) $\,$ If the space $ \underline{\cR}(\RR^s)$ is a closed
  subspace of $L^2(\RR^s)$, there exists by the uniform
  boundedness principle \cite{Yo} some constant
  $c_{\underline{\cR}}$ such that
  \be
  \omega_n(a^*(f) a(f)) \leq c_{\underline{\cR}} \, \| f \|^2 \, , \quad
  f \in  \underline{\cR}(\RR^s) \, , \ n \in \NN \, . 
  \ee
  Let $P_{\underline{\cR}}$ be the projection onto $\underline{\cR}(\RR^s)$
  and let $f_n$, $n \in \NN$, be a sequence of functions complying
  with condition \eqref{e.10}. One then obtains straightforwardly 
  \be
  \omega_n\big(a^*((1 - P_{\underline{\cR}})f_n)
  a((1 - P_{\underline{\cR}})f_n) \big)  
  \geq 
  \omega_n(a^*(f_n) a(f_n)) -
  2 ( c_{\underline{\cR}} n )^{1/2} \, ,
  \ee
  hence $\limsup_n (1/n) \, \omega_n\big(a^*((1 - P_{\underline{\cR}})f_n)
  a((1 - P_{\underline{\cR}})f_n) \big)  \geq \delta $. 
  Since the projection $(1 - P_{\underline{\cR}})$ has by assumption
  finite dimension $d$, it follows that there exists in the compact unit ball
  of the corresponding subspace some normalized function $f$ such that
  \be
\limsup_n  (1/n) \, \omega_n(a^*(f) a(f)) \geq \delta/d \, ,
  \ee
  completing the proof. \end{proof}  

\medskip
As was discussed in the preceding sections, the spaces 
$\underline{\cR}(\RR^s)$ may not be expected to be closed
in general if one proceeds to the thermodynamic limit, since
then the emerging condensates are to be described by
distributions. The method to avoid this problem is to restrict
attention to local subspaces $L^2(\bO) \subset L^2(\RR^s)$. Proceeding
from the obvious characterization of sequences of
states describing locally some growing condensate,
one may expect that the resulting local regular spaces
$\underline{\cR}(\bO)$ are closed in cases of interest and also have 
finite dimensional orthogonal complements. The preceding
results then apply accordingly. We dispense with
a discussion of the obvious details. 

\medskip
In our final result, we establish a relation between the
notions of growing condensates and of proper condensates in those
cases, where one can find a fixed function $f$ for which
condition~\eqref{e.10} is satisfied. 

\begin{proposition} \label{l.4}
 Let $\omega_n$ be a sequence of $n$-particle states
 on the algebra of
 observables $\fA(\RR^s)$ for which condition~\eqref{e.10} is satisfied for
 a fixed normalized function $f \in L^2(\RR^s)$, not
 depending on $n \in \NN$. 
 There exist limit points $\omega_{\infty}$ of this sequence,
 containing in their central decomposition a
 non-negligible set of primary components with the property that 
 all $C_0$-functions of
 $a^*(f)a(f)$ vanish in the corresponding GNS representation. In
 other words, they contain a proper condensate. 
  \end{proposition}
\begin{proof} We make use of the 
  resolvents $A_\varepsilon = (1 + \varepsilon a^*(f) a(f))^{-1}$,
  $\varepsilon > 0$, considered in Proposition~\ref{p.2.1}, and show that
  the states have weak-$*$ limit points in which the
  expectation values of these resolvents comply with  
  relation \eqref{e.2}. The statement then follows from that
  proposition.

  \medskip
  Restricting the number 
  operator $a^*(f)a(f)$ to the $n$-particle space $\cF_n$, it can be
  spectrally decomposed into 
\be
  a^*(f)a(f) \upharpoonright \cF_n =
  \sum_{k = 0}^n k \, E_{n,k} \, , \quad
  n \in \NN \, .
\ee
  By assumption there
  exists a subsequence $\omega_{n_l}$ such that,
  disregarding corrections which
  vanish in the limit of large $n_l$, 
  \be
  \sum_{k = 0}^{n_l} k \, \omega_{n_l}(E_{n_l,k})
  \geq \delta \, n_l \, ,  \quad l \in \NN \, .
  \ee
  Denoting by $[x] \in \NN_0$ the largest number which is smaller
  or equal to $x > 0$, we have
  \begin{align}
  \sum_{k = 0}^{n_l} k \, \omega_{n_l}(E_{n_l,k}) &
  \leq \sum_{k = 0}^{[\delta n_l/2] } k \, \omega_{n_l}(E_{n_l,k})
  + \sum_{k = [\delta n_l/2]}^{n_l} k \, \omega_{n_l}(E_{n_l,k}) \nonumber \\
  & \leq [\delta n_l/2] +
  n_l \! \! \! \sum_{k = [\delta n_l/2]}^{n_l} \omega_{n_l}(E_{n_l,k})
  \, .
  \end{align}
  It implies 
  \be
  \sum_{k = [\delta n_l/2]}^{n_l} \omega_{n_l}(E_{n_l,k}) \geq \delta  /2 \, , 
  \quad l \in \NN \, . 
  \ee
  Now let $\varepsilon > 0$ and let
  $D_\varepsilon \doteq  (1 - A_\varepsilon) 
  =  \varepsilon a^*(f) a(f) (1 +  \varepsilon a^*(f) a(f))^{-1} $.
  Then
  \begin{align}
  \omega_{n_l}(D_\varepsilon) & =
  \sum_{k = 0}^{n_l} \big(\varepsilon k / (1 +  \varepsilon k) \big) \, 
  \omega_{n_l}(E_{n_l,k}) 
  \geq \sum_{k = [\delta n_l/2]}^{n_l} \big(\varepsilon k / (1 +  \varepsilon k)
  \big) \, \omega_{n_l}(E_{n_l,k})  \nonumber \\
  & \geq \mbox{\Large
    $ \frac{\varepsilon \, [\delta  n_l/2]}{1 +
        \varepsilon \, [\delta n_l/2]} $} 
  \sum_{k = [\delta n_l/2]}^{n_l}  \! \! \omega_{n_l}(E_{n_l,k}) \ 
  \geq \ \mbox{\Large
     $ \frac{\delta \, \varepsilon \, [c n_l/2]}{2 \, (1 +
      \varepsilon \, [\delta n_l/2])} $}  \, .
  \end{align}
  It follows that
  $\big(1 - \limsup_l \omega_{n_l}(A_\varepsilon)\big)
    = \liminf_l \, \omega_{n_l}(D_\varepsilon) \geq \delta/2 > 0$,
  independently of the value of $\varepsilon > 0$. 
  Thus if $\omega_{\infty}$ is any weak-$*$ limit point on
  $\fA(\RR^s)$ of the sequence $\omega_{n_l}$, $l \in \NN$,
  one has \ $\limsup_{\varepsilon \searrow 0} \, 
  \omega_{\infty}(A_\varepsilon) \leq (1 - \delta/2) < 1$,
  completing the proof.
\end{proof}  
Thus the concepts introduced in the
present investigation also add to the understanding of
the properties of growing condensates in the limit of
large particle numbers. 

\section{Conclusions}
\setcounter{equation}{0}

\noindent
In the present investigation we have established concepts which allow it
to discuss the formation of condensates in an unambiguous manner
in states with a finite particle number. The basic idea is to 
proceed to the theoretical limit of infinite particle numbers and to 
determine the subspaces of wave functions that describe regular, finitely 
occupied excitations. This step is necessary in order to
clearly identify the wave functions of the proper condensates 
within the maze of infinitely occupied states which
have some overlap with them.
In those cases where the spaces of regular functions are closed and have an
orthogonal complement, the wave functions in this complement 
describe the proper condensates, consisting of infinitely many 
particles in the limit. In view of this feature we 
refer to those functions as singular. 
The notion of proper condensate  has the status of a superselection
rule. As we have seen, the presence of
proper condensates can be established by central
sequences of observables. 

\medskip
If the regular spaces are not closed,
it is an indication that the condensates in the limit
states have to be described by improper states. 
The regular spaces are then dual to these 
distributions and serve as test functions. 
One  can frequently bypass this feature by proceeding to
a local point of view. Restricting the limit states to local
observables, the improper states become normalizable, 
resulting in singular local subspaces of wave functions
in the orthogonal complement of the regular ones. 

\medskip
Having identified the relevant regular spaces and
their singular complements, one can revert to the states of
primary interest, having a finite
particle number, and base their analysis on these notions. The
algebras of observables which are sensitive to the condensate
are assigned to the spaces of singular functions, which are
finite dimensional whenever the proper condensates occupy only 
a limited number of modes. 
Their restrictions to states with a finite particle
number are then isomorphic to matrix algebras, allowing for 
a convenient analysis and manipulation of the condensate fraction. In
particular, the number of particles in the condensate can
be changed arbitrarily in these cases without
changing the total number of particles in the
system or affecting the regular observables. Moreover, one can
characterize the
onset of proper condensation in given sequences of
states by computing the expectation
values of the (local) particle number operators for the singular
excitations and comparing them with 
those of the full (local) particle number operators.

\medskip
We did not discuss here the possibility that the orthogonal complements 
of the regular functions have infinite dimension. Indeed this is
expected to happen
in interacting systems if one squeezes the particles into narrow trapping
potentials. Then also highly excited states become infinitely 
occupied. In our opinion these cases are only of secondary interest,
similarly to the case of infinite temperatures, where all levels are
infinitely occupied and no regular functions survive. We expect that
these less transparent forms of condensation can be avoided by
choosing appropriate sequences of trapping potentials, where
one may be able to establish the onset of proper condensation into a
limited number of states.  

\medskip
The notion of proper condensates has already appeared in an analysis of
non-interacting systems in arbitrary trapping potentials given in 
\cite{BaBu}. There the idea of focusing on subspaces of
regular excitations was systematically pursued
in order to identify the singular wave functions in their complement, 
describing proper condensates. In the present
investigation, we have seen that this idea is meaningful more
generally. We expect that it will also be useful in interacting
theories, where one has less \textit{a priori} information
about possible candidates for the condensate functions. 
As we have seen, the direct approach for the determination of 
functions describing proper condensates
is difficult. Whenever there are growing
condensates in a sequence of states,
there exists a maze of other functions (to be precise: of second category)
which are also infinitely occupied in the limit of large particle
numbers. In contrast, the set of regular functions
tends to become smaller in this limit. As a matter of fact,
if proper condensates appear, the regular functions are generically
meager (of first category), as is obvious if they have
a non-trivial orthogonal complement. This feature will
hopefully be useful for alternative 
existence proofs of proper condensates in
the presence of interaction. 

\vspace*{-3mm}
\section*{dedication}

\vspace*{-5mm}
\noindent This article is dedicated
to Helmut Reeh on the occasion of his 90th birthday. 

\vspace*{-3mm}
\begin{acknowledgments}

\vspace*{-5mm}
\noindent 
I would like to thank Jakob Yngvason for an intense exchange on the topic
of Bose-Einstein condensation and his numerous objections and 
critical comments on previous versions of this article. It helped me to
clear up my thoughts, even though
our views on this matter do not yet conform. I am
also grateful to Dorothea Bahns and the Mathematics Institute
of the University of G\"ottingen for their generous hospitality.  
\end{acknowledgments}

\appendix

\section*{Appendix}
\setcounter{equation}{0}

\vspace*{-3mm}
\noindent
We discuss in this appendix some aspects of the direct approach to the 
determination of condensate wave functions, based on the analysis of
one-particle density matrices. As we will see, the concept of
\textit{growing condensates} is quite subtle and depends crucially
on the fact that one deals with the exact eigenfunctions of these
matrices, assigned to their largest eigenvalues. If one
slightly relaxes this condition and only requires that there are
one-particle functions which are (almost) macroscopically occupied,
there exists an abundance of examples. So this approach does not
allow to determine in a precise manner the wave functions of
particles in a proper condensate. This observation
suggests that concentrating on their complement, the regular
functions, may also be meaningful there. 

\medskip
Turning to the analysis, let $\cH_2(\cF)$ be
the space of Hilbert-Schmidt operators on
the standard bosonic Fock space $\cF$ and let $\omega_n$
be a sequence of $n$-particle states, described by
density matrices $\rho_{n}$
on the subspaces $\cF_n$, $n \in \NN$. 
We consider the (real) linear maps
$T_n : L^2(\RR^s) \rightarrow \cH_2(\cF)$, $n \in \NN$,  given by
\be \label{e.4}
T_n(f) \doteq a(f) \, \rho_{n}^{1/2} \, , \quad f \in L^2(\RR^s) \, . \tag{A.1}
\ee
Clearly,
\be
\| T_n(g) \|_{\cH_2(\cF)}^2 = \omega_{n}(a^*(g) a(g))
\leq n \, \| g \|^2 \, ,  \quad g \in L^2(\RR^s) \, , \tag{A.2}
\ee
so these maps are bounded. The following result is then
an immediate consequence of a theorem by Banach \cite[II.4]{Yo}. 

\medskip \noindent
\textit{Observation 1:} The linear space of functions
$\underline{\cR}(\RR^s)  \subset L^2(\RR^s)$ satisfying
\be
\limsup_n \, \| T_n(f) \|_{\cH_2(\cF)} < \infty \, , \tag{A.3}
\quad f \in \underline{\cR}(\RR^s)  \, ,
\ee 
either coincides with all of $L^2(\RR^s)$, or it is meager
(a countable union of nowhere dense sets) in $L^2(\RR^s)$. 

\medskip 
In view of condition~\eqref{e.10},
the latter alternative obtains in the
presence of growing condensates. So the set of
one-particle regular functions is small. In order to
obtain some further insights into the structure of their
singular complement, let us slightly relax condition
\eqref{e.10} by replacing the factor $(1/n)$ 
by $(\chi(n)/n)$, where~$\chi$ is a function tending
arbitrarily slowly to $\infty$ in the limit of large
$n$. Functions $f \in L^2(\RR^s)$ for which the resulting condition
is satisfied are said to be almost macroscopically occupied. 
As a matter of fact, one can arbitrarily choose a countable
number of functions $\chi_m$, $m \in \NN$, with
decrescent increase and there still exists an abundance
of functions $f \in L^2(\RR^s)$ complying with the resulting conditions.
This observation is based on a result of Banach-Steinhaus
(condensation of singularities \mbox{\cite[II.4]{Yo}}).

\medskip \noindent
\textit{Observation 2:} Let 
$\omega_n$, $n \in \NN$, be a sequence of states,
exhibiting growing condensates, and let 
$\chi_m$, $m \in \NN$, be a sequence of
functions which are arbitrarily slowly tending 
to $\infty$. There exists a subset of functions $\underline{\cS}(\RR^s)$ 
of second category in the complement of the
one-particle regular functions $\underline{\cR}(\RR^s)$
  such that
  \be 
  \limsup_n \, (\chi_m(n)/n) \, \omega_n(a^*(f)a(f)) = \infty  \, ,
  \quad f \in  \underline{\cS}(\RR^s) \, , \ m \in \NN \, . \tag{A.4}
  \ee  

  \medskip
  We briefly sketch the proof of this statement. 
    Let $T_n$, $n \in \NN$, be the maps defined in \eqref{e.4}.
  We then consider the double sequence
  $T_{m,n} \doteq (\chi_m(n)/n)^{1/2} \, T_n \,$ for $m,n \in \NN$. 
  According to the preceding argument, these maps are bounded,
  $\| T_{m,n} \| \leq \chi_m(n)^{1/2} $.  
  Since the given sequence of states contains growing condensates, 
  there exist functions $f_n$, $n \in \NN$, such that
  for some subsequence of states  
  \be
  \omega_{n_k}(a^*(f_{n_k}) a(f_{n_k})) \geq (\delta/2) \, n_k \, ,
  \quad k \in \NN \, . \tag{A.5}
  \ee
  This implies, for any given $m \in \NN$,  that \ 
  $\limsup_k \| T_{m, n_k}(f_{n_k}) \|_{\cH_2(\cF)} = \infty$
  in view of the divergence of $\chi_m(n)$ in the limit
  of large $n$.
  It follows that the set of functions $M_m$, $m \in \NN$, 
  for which $\limsup_k \| T_{m, n_k}(g) \|_{\cH_2(\cF)} < \infty$,
  $g \in M_m$, is a meager set. Thus the countable union 
  $\cup_m M_m \subset L^2(\RR^s)$ is also meager,
  so its complement $\underline{\cS}(\RR^s)$ is of second category.
  
\bigskip \noindent
\textbf{\large Data availability} \\[1mm]  
Data sharing is not applicable to this article as no new data were
created or analyzed in this study.

\end{document}